\documentclass[10pt,journal]{IEEEtran}
\normalsize
\usepackage{amsthm,amssymb,amsmath,tikz,graphics,cite,float,graphicx,epstopdf,epsfig,verbatim,url}
\usetikzlibrary{automata}
\usetikzlibrary{shapes,arrows}
\newtheorem{lem}{Lemma}
\newtheorem{corol}{Corollary}
\usepackage[utf8]{inputenc}
\newtheorem{thm}{Theorem}

\newtheorem{rem}{Remark}

\newtheorem{mydef}{Definition}

\usepackage{capt-of}
\usepackage[noend]{algpseudocode}
\usepackage{algorithmicx}
\usepackage[ruled]{algorithm}

\tikzset{
  treenode/.style = {align=center, inner sep=0pt, text centered,
    font=\sffamily},
  arn_r/.style = {treenode, circle, black, draw=black, 
    text width=1.5em, very thick},
}
 
\begin{document}
\title{Throughput Analysis of Decentralized Coded Content Caching in  Cellular Networks}
\author{Mohsen~Karimzadeh~Kiskani$^{\dag}$,
        and Hamid~R.~Sadjadpour$^{\dag}$
\thanks{M. K. Kiskani$^{\dag}$ and H. R. Sadjadpour$^{\dag}$ 
are with the Department of Electrical Engineering, University of California, Santa Cruz. Email: 
\{mohsen, hamid\}@soe.ucsc.edu}}

\markboth{IEEE Transactions on Wireless Communications}{IEEE Transactions on Wireless Communications}

\maketitle \thispagestyle{empty}
\begin{abstract}
Decentralized coded content caching for next generation cellular networks is studied. The contents are linearly combined and cached in under-utilized caches of User Terminals (UTs) and its throughput capacity is compared with decentralized uncoded content caching. In both scenarios, we consider multihop Device-to-Device (D2D) communications and the use of femtocaches in the network. It is shown that decentralized coded content caching can increase the network throughput capacity compared to decentralized uncoded caching by reducing the number of hops needed to deliver the desired content. Further, the throughput capacity 
for Zipfian content request distribution is computed and it is shown that  the decentralized coded content cache placement can increase the throughput capacity of cellular networks by a factor of $(\log (n))^2$ where $n$ is the number of nodes served by a femtocache. 
\end{abstract}


\IEEEpeerreviewmaketitle

\section{Introduction}

\IEEEPARstart{R}{ecent} advances in storage technology  have made it possible for many consumer and user electronic products with  Terabyte of storage capability. Many researchers are  investigating the possibility of reusing this under-utilized storage capability to cache popular contents in order to improve the content delivery in cellular networks.  

In recent years, the problem of caching has been extensively studied. The fundamental limits of caching in broadcast channels is studied in \cite{maddah2014fundamental}. Other researchers  \cite{maddah2013decentralized,pedarsani2014online,hachem2014multi, karamchandani2014hierarchical} extended the results in \cite{maddah2014fundamental} for different scenarios in broadcast channels. The common features of all these studies are the assumptions that contents are cached without any coding and it is one hop communications. The authors in \cite{jeon2015caching,jeon2015wireless} analyzed the capacity of multihop networks but they still assumed contents are cached without any coding, i.e., uncoded caching. Further, these studies \cite{jeon2015caching,jeon2015wireless} focus on wireless ad hoc networks and there is no extension of the work to cellular networks.

In this paper, we propose a radically  different  cache placement approach. While in our proposed algorithm each UT caches independently of all other UTs in a decentralized manner, redundant caching is avoided by storing a random bitwise XOR combination of popular contents. We call this method {\em decentralized coded content cache placement} algorithm. This approach  increases the network throughput capacity and does not suffer from over-caching problem of uncoded caching.

The proposed coded caching  is fundamentally different from the notion of coded caching in references like
\cite{maddah2013decentralized,pedarsani2014online,hachem2014multi, karamchandani2014hierarchical}. In such papers, during the cache placement phase, only uncoded contents or uncoded parts of contents are stored in the caches. Later during the content delivery phase, the base station broadcasts coded contents (linear combination of multiple contents) to UTs such that they can decode their files simultaneously. We instead propose that during the cache placement phase, the contents are randomly combined and cached in UTs. 
It is shown that this  coded caching approach performs near optimal in terms of the average  number of hops to retrieve a content and hence, it can significantly increase the network throughput capacity. This makes 
the proposed coded cache placement very suitable in practical systems where UTs have small storage capability compared to the total number of contents in the network. 

Many studies propose to utilize  high bandwidth D2D communications  for UTs. Current IEEE 802.11ad standard \cite{ieee80211ad} and the millimeter-wave proposal for future 5G networks \cite{boccardi2014five,rappaport2013millimeter} are examples of such  high bandwidth D2D communications.
Authors in \cite{kiskani2015multihop} extended the solution in \cite{golrezaei2012femtocaching}
to deliver the contents from the helpers to the UTs using multihop D2D  communications.  However, \cite{kiskani2015multihop} only considers uncoded caching.

We study our approach within the framework of future cellular networks that use femtocaches (or helpers) \cite{golrezaei2012femtocaching}. In such networks, several {\em helpers} with high storage capabilities 
are deployed in each cell to create a distributed wireless caching infrastructure. Each helper is serving a
wireless ad hoc network of UTs through multihop D2D communications. We assume that helpers are connected to the base station through a high bandwidth backhaul infrastructure. For simplicity of our analysis, we assume that the contents have equal sizes. Our results are valid for contents with different sizes since
in practice each content can be divided into equal chunks. We prove that the proposed  decentralized coded content caching  increases the capacity of cellular networks by a factor of $(\log(n))^2$ compared to decentralized uncoded  caching. As far as we know, this is the first paper to propose the idea of decentralized coded content caching and to prove that  coded cache placement can  increase the network  capacity.

The rest of the paper is organized as follows. In section \ref{relwork}, the related work is discussed and  
section \ref{netmod} describes the network model considered in this paper. Section \ref{uncoded_sec} focuses on the computation of the throughput capacity of wireless cellular networks operating under a decentralized uncoded cache placement algorithm and  section \ref{coded_sec} reports the  capacity  of 
coded cache placement algorithm. In section \ref{zipfian_networks}, we compute the  capacity of networks operating under a Zipfian content request distribution. Simulation results are shown in section \ref{sim_sec}. Section \ref{disc_sec} compares this work with other coding schemes and the paper is concluded in section \ref{conc_sec}. 

\section{Related Work}
\label{relwork}
The original femtocache network model \cite{golrezaei2012femtocaching,shanmugam2013femtocaching} was focused on delivery of contents from femtocaches to UTs using single hop communications. The authors in \cite{kiskani2015multihop} considered a femtocaching  network with multihop D2D relaying of information from the helper to the UTs. A solution based on index coding was proposed in which the helper utilizes the side information in the UTs to create index codes which are then multicasted to the UTs. The approach  reduces bandwidth utilization by grouping multiple unicast transmissions into multicast transmission. 
While \cite{kiskani2015multihop} proposed a solution for the helpers to efficiently multicast the contents to the UTs, this paper assumes that the helper only unicasts the contents to the UTs. The UTs cache uncoded 
 contents in \cite{kiskani2015multihop} while in this paper a decentralized coded content caching solution for UTs is proposed.

Caching has  been a subject of interest to many researchers. The fundamental information theoretical
limits of caching is studied in \cite{maddah2014fundamental} where the authors studied the problem of caching in broadcast channels with a central uncoded cache placement algorithm. The authors in \cite{maddah2013decentralized} extended the work of \cite{maddah2014fundamental} to distributed uncoded cache placement approach and then broadcasting coded contents during the delivery phase over the shared link. They \cite{maddah2013decentralized} proposed  to break the contents into parts and then the UTs randomly cache the content parts during placement phase. In the delivery phase, coded contents are broadcasted from the server such that the UTs can decode their desired contents optimally. This paper proposes to randomly and independently combine contents and store them during cache placement phase.
During the delivery phase, a linear combination of encoded files is used to retrieve the requested 
content. The notion of coded caching in \cite{maddah2014fundamental} and all the papers that followed \cite{niesen2014coded,pedarsani2014online,karamchandani2014hierarchical,hachem2014multi}
refers to broadcasting coded contents during the delivery phase and it is not a cache placement technique. All prior works  \cite{maddah2014fundamental, maddah2013decentralized,pedarsani2014online,hachem2014multi,
karamchandani2014hierarchical, niesen2014coded} are fundamentally different from this paper as they are studying the information theoretic bounds of caching of a single server connected to  users through a shared link while this paper studies the scaling behavior of networks in which the UTs retrieve requested contents  through multiple hops.

Other papers  \cite{ji2014fundamental,ji2013wireless,jeon2015wireless, jeon2015caching} studied the problem of caching in wireless and D2D networks. Authors in \cite{ji2013wireless} discussed the fundamental capacity of D2D communication. In \cite{ji2014fundamental}, a single hop D2D caching system is studied from an information theoretic point of view. The authors in \cite{jeon2015caching,jeon2015wireless} have studied the capacity of multihop wireless D2D ad hoc networks with uncoded caching in certain regimes. However, our work is essentially different from \cite{jeon2015caching,jeon2015wireless} in the sense that the UTs in our current paper always request the contents from the helper while in \cite{jeon2015caching,jeon2015wireless}, a wireless ad hoc network is considered. Clearly, such network model requires higher overhead to locate the route to the requested content while in our approach, the request always is sent toward the helper. They find the capacity for specific regimes where the cache size is relatively very large compared to the number of UTs. In this paper, we prove that even constant cache size provides considerable capacity gain. Another major difference between our work and the references \cite{ji2014fundamental,ji2013wireless,jeon2015wireless, jeon2015caching} is the introduction of decentralized coded content caching which has not been studied in these works.

Caching coded contents has been previously suggested \cite{lee2015index,chen2014fundamental} as an efficient caching technique for devices with small storage capacity. In \cite{lee2015index}, the problem of index coding with coded side information is studied and \cite{chen2014fundamental} proposed a coded caching strategy for systems with small storage capacity. Our results demonstrate that apart from the practical importance of coded caching benefits for  small storage devices, it can also increase the throughput capacity of cellular networks. 

An earlier version of this paper was presented at \cite{kiskani2016capacity} which did not discuss the congestion problem and Zipfian content request distribution and details of some proofs are missing.

\section{Network Model}
\label{netmod}
In this paper, we  study the  capacity of cellular networks utilizing a distributed femtocaching infrastructure as proposed in \cite{golrezaei2012femtocaching}. We assume several helpers with high storage capabilities are deployed throughout the network to assist in delivering the contents through multiple hops to UTs.

For capacity analysis, we  use the deterministic routing approach proposed in \cite{kulkarni2004deterministic}. Without loss of generality, it is assumed that the UTs are distributed on a square of area one and the helper is located at the center serving $n$ UTs which are  randomly distributed on a  square. The square is divided into many square-lets of side length $c_1 s(n)$ where  $s(n) \triangleq \sqrt{\frac{\log (n)}{n}}$. It is proved  \cite{penrose1997longest} that this network is connected if nodes  have a transmission range of $\Theta \left(s(n)\right)$. When a UT  requests a content from the helper, the content is routed \cite{kulkarni2004deterministic} from the helper to the UT in a sequence of horizontal and vertical straight lines through square-lets which connect the helper to the UT.

A {\em Protocol Model} is considered \cite{xue2006scaling} for  successful communication between UTs. 
According to this model, if the UT $i$ is located at   $Y_i$, then a transmission from $i$ to  UT $j$ is
successful if $|Y_i-Y_j | < s(n)$ and for any other UT $k$ transmitting on the same frequency band, $|Y_k-Y_j| > (1 + \Delta)s(n)$ for a fixed guard zone factor $\Delta$. A Time Division Multiple Access (TDMA) scheme is assumed for the transmission between the square-lets. With the assumption of Protocol Model,
then the square-lets with a distance of $c_2=\frac{2+\Delta}{c_1}$ square-lets apart can transmit simultaneously without significant interference \cite{kulkarni2004deterministic}.

The results are computed in terms of scaling laws. We use \cite{knuth1976big} the following order notations. Denote $f(n)=\operatorname{O}(g(n))$ if there exist $c>0$ and $n_0>0$ such that $f(n)\leq c g(n)$ for all $n\geq n_0$,  $f(n)=\Omega(g(n))$ if $g(n)=\operatorname{O}(f(n))$, and $f(n)=\Theta(g(n))$ if $f(n)=\operatorname{O}(g(n))$ and $g(n)=\operatorname{O}(f(n))$.

The contents in the network  are  represented by a set  $\mathcal{X}=\{x_1,x_2,\dots,x_m\}$ and the set of 
their indices by $\xi=\{1,2,\dots,m\}$. Without loss of generality, we assume that the contents with lower indices are more popular than the ones with higher indices. The contents are categorized into two groups of
popular  and less popular contents.
\begin{mydef}
{\em 
Define the set of $h$ most popular contents as $\mathcal{X}_{h} = \{x_1,x_2,\dots,x_h\} \subseteq \mathcal{X}$ where $\xi_h = \{1,2,\dots,h\} \subseteq \xi$ denotes the set of indices of the most popular contents. 
} \label{def_popular}
\end{mydef}

The number of cached popular contents  during the cache placement phase, $h$, is decided by the cellular network designer based on different parameters and specifications of the network. The selection of $h$ is critical in the frequency of broadcast of unpopular contents by the base station. Typically $h$ is chosen large enough such that with a very low probability contents are broadcasted from the base station. We assume that $m$ and $h$ grow polynomially with $n$ similar to \cite{jeon2015caching,jeon2015wireless}. Since $h$ is a small fraction of $m$, we assume that $h$ is growing with $n$ in a much slower rate compared to $m$.  Section \ref{zipfian_networks} describes the necessary growth rate of $h$ to guarantee that the probability of requesting a content with index larger than $h$ decays polynomially with $n$ with a decay rate of $\rho$. The results are general in nature because by allowing  $h$ and $m$ scale with $n$ with different values of exponents, all possible values of $h$ and $m$ are considered.

The UTs have the same  cache  size of $M$ and $M < h$\footnote{Otherwise, the maximum throughput capacity
is trivially achievable by caching all the popular contents in each UT.}. There is no restriction on the cache size $M$ and $M$ is a constant or a function of $n$ as in \cite{jeon2015caching,jeon2015wireless}.
During the cache placement phase, UT caches are filled independently of other UTs. 

Helpers are assumed to store all the popular contents in $\mathcal{X}_h$. The popular content requests are served by D2D multihop communications and the less popular content requests are served by the base station through the low bandwidth shared link. The achievable throughput and network capacity are defined as follows.

\begin{mydef}{\em 
 A network throughput of $\lambda(n)$ contents per second for each UT is {\em achievable} if there is a scheme for scheduling transmissions in the D2D multihop network, such that every popular content request from $\mathcal{X}_h$ by every UT at a rate of $\lambda(n)$ can be served by the D2D multihop network.
 }\label{def_achievable}
\end{mydef}

\begin{mydef}{\em 
The {\em throughput capacity} of the network is lower bounded by $\Omega(g(n))$ contents per second if a deterministic constant $c_3 > 0$ exists such that 
\begin{equation}
   \lim_{n \to \infty} ~\mathbb{P}[\lambda(n) = c_3 g(n) ~\textrm{is achievable}~] = 1. 
 \label{lb_def_cap}
 \end{equation}
The network throughput capacity is upper bounded by $O(g(n))$ contents per second if a deterministic constant  $c_4 < + \infty$ exists such that 
\begin{equation}
 \liminf_{n \to \infty} ~\mathbb{P}[\lambda(n) = c_4 g(n) ~\textrm{is achievable}~] <1.
 \label{ub_def_cap}
\end{equation}
The network throughput capacity is of order $\Theta(g(n))$ contents per second if it is lower bounded by $\Omega(g(n))$ and upper bounded by $O(g(n))$.
 }\label{def_cap}
\end{mydef}

This paper assumes that the cache placement is already done and we want to study the throughput capacity during the content delivery phase. If the content can be decoded using the cached information in the intermediate relaying UT caches, then the UT does not need to receive the content from the helper. However, if the content cannot be decoded using the intermediate relays, then the content is received from the helper. 

To simplify the analysis, all contents are assumed of equal size with each having $Q$ bits. This is a reasonable assumption since in practice the contents are divided into equal-sized chunks. The minimum number of hops required to successfully decode any content is denoted by $Y$ with the average value of $Y$ 
 taken over all possible content requests denoted by $\mathbb{E}[Y]$. If the maximum achievable network throughput is $\lambda(n)$, then the network can deliver $n \lambda(n)$ contents per second or  equivalently, UTs can transmit $n \lambda (n) \mathbb{E}[Y] Q$ bits per second. There are exactly 
$\frac{1}{(c_2 c_1 s(n))^2}$ square-lets at any time slot available for transmission. The maximum number of bits that the network can deliver is equal to $\frac{W}{(c_2 c_1 s(n))^2}$ where $W$ is the total available bandwidth. Therefore,
\begin{equation}
\label{capa}
 \lambda(n) = \dfrac{W}{n  \mathbb{E}[Y] Q (c_2 c_1 s(n))^2} = \Theta 
 \left(\dfrac{1}{\mathbb{E}[Y] \log n} \right).
\end{equation}
Hence, to compute the maximum achievable network throughput, it is enough \cite{azimdoost2013} to 
find the average number of transmission hops needed to deliver the popular contents. Let's denote the requested content index by $r$, the probability of requesting $i^{th} $content by $f(i) = \mathbb{P}[r=i]$ and the cumulative probability function by $F(i) = \mathbb{P}[r \le i]$. This implies that 
\begin{equation}
 \mathbb{P}[r \in \xi_h] = \mathbb{P} [r \le h] = F(h).
 \label{ph_emn}
\end{equation}
With uncoded caching, if a UT $\mathcal{U}$ requests a content, the content is delivered to $\mathcal{U}$ either by the helper or by a relay that caches the content and is located on the routing path between $\mathcal{U}$ and the helper. With coded caching, if a group of coded contents  can be used to decode the content are available along the routing path between  $\mathcal{U}$ and the helper, then the helper informs the UTs the decoding instructions. If sufficient  coded files do not exist in the caches of the UTs between $\mathcal{U}$ and helper to decode the desired content, then the content is routed to $\mathcal{U}$ from the helper. Since the helper only stores the popular contents in $\xi_h$ and the less popular contents are
downloaded from the base station, then the average  number of traveled D2D hops in the network can be written as 
 \begin{eqnarray}
 \mathbb{E}[Y] = \mathbb{E}[Y | r \in \xi_{h}] \mathbb{P}[r \in \xi_{h}] = \mathbb{E}[Y | r \in \xi_{h}] F(h).
 \label{ex113}
\end{eqnarray}

For many web applications \cite{breslau1998implications,breslau1999web}, the  content request popularity follows Zipfian-like distributions. Although we  express our results in general form, we will later compute explicit capacity results assuming a Zipfian content popularity distribution. Our main result in proving the gain of coded caching over uncoded caching is independent of  the  content popularity distribution. 

For Zipfian popularity distribution with parameter $s$, the probability of requesting a content with popularity index $i$ is
\begin{equation}
 f(i)=\mathbb{P}[r = i] = \frac{i^{-s}}{\sum_{j=1}^m j^{-s}} = \frac{i^{-s}}{H_{m,s}},
 \label{mgfbf}
\end{equation}
where $H_{m,s}$ represents the generalized harmonic number with parameter $s$. The rest of the paper is dedicated to computing the throughput capacity of both decentralized uncoded and coded content caching including for special case of Zipfian content popularity distribution.

\section{Decentralized Uncoded Content Caching}
\label{uncoded_sec}
This section focuses on throughput analysis in cellular networks when each UT cache $M$ popular contents drawn uniformly at random independently of all other UTs. The uniform distribution of cache placement is different from the content request distribution by UTs. We will study the network throughput assuming a fixed cache placement. 
\begin{lem}
 {\em If a content is requested independently and uniformly at random from $h$ most popular contents in  $\mathcal{X}_{h}$, then the average required number of requests to have at least one copy of each content is equal to 
 \begin{equation}
  \mathbb{E}_{\textrm{coupon collector}} = h \sum_{i=1}^{h} \frac{1}{i} = h H_{h} 
  = \Theta(h \log(h)),
  \label{codhf}
 \end{equation}
 where $H_h = \Theta(\log(h))$ is the $h^{th}$ harmonic number.
 }
 \label{leme0}
\end{lem}
\begin{proof}{This is known  as the coupon collector problem \cite{erdHos1961classical}.}\end{proof}
\begin{lem}{\em
If each UT caches $M$ different contents uniformly at random during cache placement phase, then the average number of UTs required to have at least one copy of each content in the network is 
  \begin{align}
   \dfrac{h H_h}{d(h,M)} \le \mathbb{E}_{\textrm{uncoded}} \le 1+\dfrac{h H_h}{d(h,M)},
  \end{align}
where 
  \begin{equation}
  d(h,M) \triangleq \sum_{j=0}^{M-1} \frac{h}{h-j}.
  \label{eq_lem_batch}
 \end{equation}
}\label{lem_batch}
\end{lem}
\begin{proof}
This  is the extension of the coupon collector problem because the cached contents in each UT are different. To compute the average number of UTs in this case, we start from a classic coupon collector  problem. Assume that contents are chosen uniformly at random and as soon as $M$ different contents are found, they are cached in the first UT. Then the same process starts over for the next UT and after finding $M$ different contents, the contents are placed in the UT's cache. Assume that this process is repeated until  one copy of each content is cached in at least one UT's cache. Since this is a geometric distribution, on average we need $d(h,M)$ content requests to fill up one UTs cache with $M$ different contents. Based on Lemma \ref{leme0}, after an average  of $h H_h$ content requests, we  have requested
 one copy of all contents. Hence, the average number of UTs required to have one copy of each content in at least one UT cache is between $\dfrac{h H_h}{d(h,M)}$ and $1+\dfrac{h H_h}{d(h,M)}$.
\end{proof}

\begin{thm}{\em 
If $ h \log(h) = O(M {s(n)^{-1}})$\footnote{This condition means that the average number of hops needed to find the content is less than the number of hops to the helper. If this condition does not hold, then the content is sent by the helper to the requesting UT.}, then the average number of transmission hops to receive the contents in uncoded caching is equal to
  \begin{equation}
  \mathbb{E}[Y | r \in \xi_{h}] = 
  \Theta \left(\frac{h \log(h)}{M} \right).
  \label{ex_uncoded}
 \end{equation}
  \label{thm_uncoded}
}\end{thm}
\begin{proof}
 Lemma \ref{lem_batch} shows that the average number of UTs needed so that all of the requests can be satisfied is  
 \begin{align}
\dfrac{h H_h}{d(h,M)} \le \mathbb{E}[Y | r \in \xi_{h}] \le 1+\dfrac{h H_h}{d(h,M)}.  
 \end{align}
 Therefore, for large values of $h$, the average number of UTs required for finding any content scales as 
 \begin{align}
 \mathbb{E}[Y | r \in \xi_{h}] =  \Theta \left(\dfrac{h H_h}{d(h,M)} \right).
 \label{eqs_thm_1_2}
 \end{align}
 To find a bound for $d(h,M)$, notice that the series in the right hand side of equation \eqref{eq_lem_batch} has $M$ terms and the maximum and minimum values  are  $\frac{h}{h-M+1}$ and 1 respectively. Therefore, $d(h,M)$ lower and upper bounds are  
 \begin{equation}
  M \le d(h,M) \le \frac{M h}{h-M+1}.
  \label{eq_lab}
 \end{equation}
 Using equations \eqref{eqs_thm_1_2} and \eqref{eq_lab} we can find upper and lower bounds on 
 $\mathbb{E}[Y | r \in \xi_{h}]$ as 
 \begin{align}
  \mathbb{E}[Y | r \in \xi_{h}] &= O \left( \frac{h \log(h)}{M}\right), 
  \label{eqs_ub_thm1} \\
  \mathbb{E}[Y | r \in \xi_{h}] &= \Omega \left(\frac{h \log(h) (h-M+1)}{M h} \right). 
  \label{eqs_lb_thm1}
  \end{align}
Assuming $h>>M$, we have $h-M+1 \approx h$ and  the lower bound in \eqref{eqs_lb_thm1} becomes equal to the upper bound \eqref{eqs_ub_thm1} which proves the theorem.  
\end{proof}
 Figure \ref{fig_pffrk} shows that, each request can be satisfied by another UT on average $\Theta({h \log(h)}/{M})$ hops away provided that $ h \log(h) = O(M {s(n)^{-1}})$. Theorem \ref{thm_uncoded} and Equations \eqref{capa} and \eqref{ex113} can be used to prove the following corollary.

\begin{figure}[http]
\centering
\begin{tikzpicture}[->,>=stealth',shorten >=1pt,auto,node distance=1.5cm, semithick]
  \node[rectangle,draw]  (N0) {};
  \node[circle,draw]  (N1) [right of=N0] {};
  \node[circle,draw]  (N2) [right of=N1] {};
  \node[circle,draw]  (N3) [right of=N2] {};
  \node[circle,draw]  (N4) [right of=N3] {};
  \node[circle,draw]  (Nl) [right of=N4] {};
  \node (TT3)[above of = N0, yshift=-1cm] {\footnotesize{H}};
  \node (T3) [above of = N1, yshift=-1cm] {\footnotesize{UT$_l$}};
  \node (T4) [above of = N2, yshift=-1cm] {\footnotesize{UT$_{l-1}$}};
  \node (T5) [above of = Nl, yshift=-1cm] {\footnotesize{UT$_0$}};
  \node (T6) [above of = N3, yshift=-1cm] {\footnotesize{UT$_2$}};
  \node (T7) [above of = N4, yshift=-1cm] {\footnotesize{UT$_1$}};
  \path (N0)  [dashed] edge node {} (N1);
  \path (N1)   edge node {} (N2);
  \path (N2)  [dashed] edge node {} (N3);
  \path (N3)  edge node {} (N4);
  \path (N4)  edge node {} (Nl);
\end{tikzpicture}
\caption{UT$_0$ is requesting a content which is available in another UT $l$ hops away along the  path toward helper H. }
\label{fig_pffrk}
\end{figure}
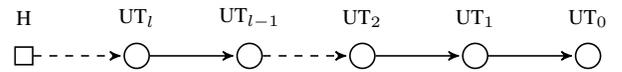
\begin{corol}{\em The throughput capacity of the decentralized uncoded content caching network with $ h \log(h) = O(M {s(n)^{-1}})$ is upper bounded by
    \begin{equation}{\color{black}
  \lambda_{\textrm{uncoded}}(n) = O \left( \frac{M}{h \log(h)  F(h) \log n} \right).
  }\label{capa_uncoded93}
 \end{equation}
The throughput capacity in equation \eqref{capa_uncoded93} is an upper bound and cannot be achieved because of congestion. The  achievable throughput capacity will be computed in the following section.
 \label{corol_uncoded56}
}\end{corol}

\section{Decentralized Coded Content Caching}
\label{coded_sec}
The  throughput capacity of decentralized coded content caching is computed in this section. We propose a random coding strategy and prove that the throughput capacity of the network is increased compared to uncoded caching strategy. 

{\bf Coded cache placement:}  In this paper, we assume the files are binary and all operations are in GF(2). For each encoded file, the helper randomly selects each one of the contents from the set $\mathcal{X}_{h}$ with probability $\frac{1}{2}$ and then combines all the selected contents  (XOR) to create one encoded file. The encoded cached content at the $j^{th}$ cache location of UT $i$ can be represented as 
 \begin{equation}
  f_j^i = \sum_{k=1}^{h} a_k^{ij} x_k = {\bf v}_j^i \mathbf{X},
  \label{f_i}
 \end{equation}
 where $\mathbf{X}=[x_1 ~ x_2 ~ \dots ~ x_h]^T$ is a column vector containing all popular contents of set $\mathcal{X}_{h}$ and ${\bf v}_j^i$ is a  uniformly distributed encoding vector with binary elements and the summation is carried over GF(2). For a UT with cache size $M$, the helper creates $M$ such encoded files. Therefore, each one of the contents in $\mathcal{X}_{h}$ has been used on average $\frac{M}{2}$ times to create the $M$ coded files. We will represent the coded contents in caches as vectors belonging to $\mathbb{F}_2^{h}$.

{\bf Coded file reconstruction:} The UT sends the request for a content to the helper. The helper then decides to send the file through a routing path as proposed in \cite{kulkarni2004deterministic}. However, it is highly possible that the content can be reconstructed using a linear combination of some coded files in the caches of UTs between the requesting UT and the helper along the routing path. If these encoded files contain $h$ linearly independent encoded vectors, they can span the entire message space. As depicted in Figure \ref{fig_pffrk22}, UT$_i$ in the routing path can contribute up to $M$ linearly independent  vectors ${\bf v}_1^i,{\bf v}_2^i,\dots,{\bf v}_M^i$ for decoding  of a content. Therefore, UT$_i$ which is at most $q$ hops away from UT$_0$ on the routing path applies gain $b_j^i \in \textrm{GF}(2)$ to its $j$-th cache content and then passes it down to the next hop closer to UT$_0$. This process of relaying and clever use of the caching contents continues hop by hop until the file reaches the requesting UT. After the requesting UT receives $ (\sum_{i=1}^q \sum_{j=1}^M b^i_j {\bf v}^i_j) \mathbf{X}$, it can reconstruct its desired content by applying its own coding gains to construct $(\sum_{i=0}^q \sum_{j=1}^M b^i_j {\bf v}^i_j) = e_k$ where $e_k$ is a vector with all elements equal to zero except the $k^{th}$ element, i.e., $x_k$  is reconstructed.

In this distributed decoding scheme, each relay UT adds some encoded files to the received file and relay it to the next hop. The coefficients $b_j^i \in \{0,1\}$ are selected such that the linear combinations of encoded files produce the desired requested content. Each relay UT that participates in this distributed decoding operation, receives $M$ bits from helper in order to combine its cached encoded files. The computational complexity for each UT is modest since it only involves XOR operation. The following lemma computes the average number of vectors ${\bf v}_j^i$ to create $h$ linearly independent vectors.
\begin{figure}[http]
\centering
\begin{tikzpicture}[->,>=stealth',shorten >=1pt,auto,node distance=1.5cm, semithick]
  \node[rectangle,draw]  (N0) {};
  \node[circle,draw]  (N1) [right of=N0] {};
  \node[circle,draw]  (N2) [right of=N1] {};
  \node[circle,draw]  (N3) [right of=N2] {};
  \node[circle,draw]  (N4) [right of=N3] {};
  \node[circle,draw]  (Nl) [right of=N4] {};
  \node (TT3)[above of = N0, yshift=-1cm] {\footnotesize{H}};
  \node (T3) [above of = N1, yshift=-1cm] {\footnotesize{UT$_q$}};
  \node (T4) [above of = N2, yshift=-1cm] {\footnotesize{UT$_{q-1}$}};
  \node (T5) [above of = Nl, yshift=-1cm] {\footnotesize{UT$_0$}};
  \node (T6) [above of = N3, yshift=-1cm] {\footnotesize{UT$_2$}};
  \node (T7) [above of = N4, yshift=-1cm] {\footnotesize{UT$_1$}};
  \path (N0)  [dashed] edge node {} (N1);
  \path (N1)   edge node {} (N2);
  \path (N2)  [dashed] edge node {} (N3);
  \path (N3)  edge node {} (N4);
  \path (N4)  edge node {} (Nl);
\end{tikzpicture}
\caption{Each requested content by UT$_0$ is constructed by a linear combinations of the contents in $q+1$ UTs caches on the path between the helper and UT$_0$. }
 \label{fig_pffrk22}
\end{figure}
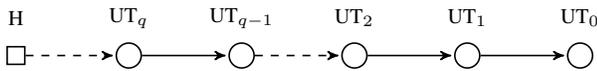  
\begin{lem}
 {\em 
 Let ${\bf v}^i_j$ be a random vector belonging to $\mathbb{F}_2^{{h}}$ with binary elements  with uniform distribution. The average number of vectors ${\bf v}^i_j$ to span the whole space of $\mathbb{F}_2^{{h}}$ equals  
 \begin{equation}
  \mathbb{E}_v = {h} + \sum_{i=1}^{h} \frac{1}{2^{i}-1}  =   {h} + \gamma,
\label{gamma}
 \end{equation}
 where $\gamma$ asymptotically approaches  the Erdős–Borwein constant ($\approx 1.6067$)\footnote{In our problem, $h$ is large enough that we can approximate the summation in \eqref{gamma} with asymptotic value.}.
 }
 \label{leme1}
\end{lem}
\begin{proof} 
 We use a Markov chain to model the problem. The states of this Markov chain are equal to the dimension of the space spanned by vectors\footnote{For the rest of lemma, we drop superscript from notations where it is obvious.} $v_1,v_2,\dots,v_l$. Let $k_l$ ($k_l \leq h$) denote the dimension of the space spanned by these vectors. Therefore, the Markov chain will have $k_l+1$ distinct states. Assuming that we are in state $k_l$, we want to find the probability that adding a new vector will change the state to $k_l+1$. When we are in state $k_l$, $2^{k_l}$ vectors out of $2^{h}$ possible vectors will not change the dimension while adding any one of $2^{h} - 2^{k_l}$ new vectors will change the dimension to $k_l+1$. Therefore, the Markov chain can be represented by the one in Figure \ref{markovchain}. 
 \begin{figure}
\begin{center}
\begin{tikzpicture}[->,>=stealth',auto,semithick,node distance=1.8cm]
\tikzstyle{every state}=[fill=white,draw=black,thick,text=black,scale=1]
\node[state]    (k0)                {};
\node[state]    (k1)[right of=k0]   {};
\node[state]    (k2)[right of=k1]   {};
\node[state]    (k3)[right of=k2]   {};
\node[state]    (kn)[right of=k3]   {};
\node [label={[label distance=0.5cm]$k_l=0$}] (t0)[below of=k0]{};
\node [label={[label distance=0.5cm]$k_l=1$}]   (t1)[below of=k1]   {};
\node [label={[label distance=0.5cm]$k_l=2$}]   (t2)[below of=k2]   {};
\node [label={[label distance=0.5cm]$k_l=3$}]   (t3)[below of=k3]   {};
\node [label={[label distance=0.5cm]$k_l={h}$}]   (tn)[below of=kn]   {};
\path (k0)    edge[loop above]    node{$\frac{1}{2^{h}}$}      (k0);
\path (k0)    edge[above]    node{$1-\frac{1}{2^{h}}$}    (k1);
\path (k1)    edge[loop above]    node{$\frac{2}{2^{h}}$}      (k1);
\path (k1)    edge[above]    node{$1-\frac{2}{2^{h}}$}    (k2);
\path (k2)    edge[loop above]    node{$\frac{2^2}{2^{h}}$}    (k2);
\path (k2)    edge[above]    node{$1-\frac{2^2}{2^{h}}$}  (k3);
\path (k3)    edge[loop above]    node{$\frac{2^3}{2^{h}}$}    (k3);
\path (k3)    edge[dashed]        node{}                     (kn);
\path (kn)    edge[loop above]    node{$1$}                  (kn);
\end{tikzpicture}
\end{center}
\caption{The state space of the Markov chain used in proof of Lemma \ref{leme1}.}
\label{markovchain}
\end{figure}
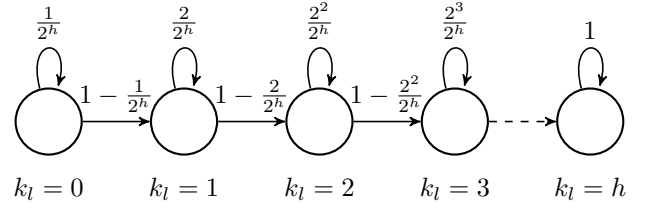 
 The state transition matrix for this Markov chain is 
\begin{equation} 
P=\begin{bmatrix} 
\frac{1}{2^{h}} & 1- \frac{1}{2^{h}} & 0                & \cdots & 0         &0        \\ 
0             & \frac{2}{2^{h}}    & 1- \frac{2}{2^{h}} & \cdots & 0         &0        \\ 
0             &   0              & \frac{2^2}{2^{h}}    & \cdots & 0         &0        \\ 
\vdots        & \vdots           & \vdots           & \ddots & \vdots    & \vdots \\
0             & 0                & 0                & \cdots           & 1- \frac{2^{{h}-2}}{2^{h}} & 0\\ 
0             & 0                & 0                & \cdots           &\frac{2^{{h}-1}}{2^{h}} & 1-\frac{2^{{h}-1}}{2^{h}}\\
0&0&0&\cdots&0& 1\\ \end{bmatrix}, \nonumber 
\end{equation}
which can be written in the form of a discrete phase-type distribution as
\begin{equation} 
P=\begin{bmatrix} 
T & T_0  \\ 
\underline{0} & 1 \\ \end{bmatrix},
\label{msdbf}
\end{equation}
where
\begin{equation} 
T=\begin{bmatrix} 
\frac{1}{2^{h}} & 1- \frac{1}{2^{h}} & 0                & \cdots & 0         &0               \\ 
0             & \frac{2}{2^{h}}    & 1- \frac{2}{2^{h}} & \cdots & 0         &0                \\ 
0             &   0              & \frac{4}{2^{h}}    & \cdots & 0         &0                 \\ 
\vdots        & \vdots           & \vdots           & \ddots & \vdots    &\vdots   \\
0             & 0                & 0                & \cdots           & \frac{2^{{h}-2}}{2^{h}} & 1- \frac{2^{{h}-2}}{2^{h}} \\ 
0             & 0                & 0                & \cdots           & 0         &\frac{2^{{h}-1}}{2^{h}}  \\
\end{bmatrix},
\end{equation}
and
\begin{equation} 
T_0=\begin{bmatrix} 
0        \\ 
\vdots \\
0\\ 
1-\frac{2^{{h}-1}}{2^{h}}\\
\end{bmatrix}.
\end{equation}
If $e$ denotes  all one vector of size ${h}$, since $P$ is a probability distribution we  have 
\begin{equation}
P\begin{bmatrix}e \\ 1 \end{bmatrix}=\begin{bmatrix}e \\ 1 \end{bmatrix}.
\label{siiiin}
\end{equation}
This 
implies that $Te + T_0 = e$, hence $ T_0 =  (I - T)e$. Therefore, it is easy to show by induction that the state transition matrix in $l$ steps can be written as
\begin{equation} 
P^l=\begin{bmatrix} 
T^l & (I-T^l)e  \\ 
0 & 1 \\ \end{bmatrix}.
\label{mskit}
\end{equation}
This equation implies that if we define the absorption  time\footnote{In our problem, absorption time is actually the total number of required vectors in relays to span the $h$-dimensional space.} as
\begin{equation}
t_a = \inf \{l \ge 1 ~|~ k_l={h}\},
\end{equation}
and if $l$ is strictly less than the absorption time, the probability of transitioning from state $i$ to state $j$ by having $l$ new vectors can be computed from   $T^l$. In other words, 
\begin{equation}
 \mathrm{P}_i^l[k_l = j,  l < t_a] = (T^l)_{ij}.
 \label{probres}
\end{equation}
Therefore, starting from state $i$, if $t_j^i$ denotes the number of vectors (i.e., encoded files for our problem) to transition from state $j$ before absorption, $t_j^i$ can be written as 
\begin{equation}
t_j^i = \sum_{l=0}^{t_a-1} \mathrm{1}\{k_l=j\}.
\label{msdhlsdjfusdg} 
\end{equation} 
The average value of $t_j^i$ is  
\begin{eqnarray}
 \mathbb{E}[t_j^i] = \mathbb{E} \left[\sum_{l=0}^{t_a-1}
 \mathrm{1}\{k_l=j\} \right] = \sum_{l=0}^{t_a-1} \mathbb{E} \left[ \mathrm{1}\{k_l=j\} \right]. 
 \label{expected1}
\end{eqnarray}
Since $\mathbb{E} \left[ \mathrm{1}\{k_l=j\} \right] = 
\mathrm{P}_i^l(k_l = j, l \le t_a-1) $, we have 
\begin{eqnarray}
\mathbb{E}[t_j^i] &=& \sum_{l=0}^{t_a-1} \mathrm{P}_i^l(k_l = j, l \le t_a-1)  \nonumber \\
 &\stackrel{a}{=}& \sum_{l=0}^{\infty} \mathrm{P}_i^l(k_l = j, l < t_a)  \nonumber \\
 &\stackrel{b}{=}& \sum_{l=0}^{\infty} (T^l)_{ij}.
 \label{expected2}
\end{eqnarray}
Equality ($a$) is correct because the probability is nonzero up to $l=t_a-1$ terms and ($b$) is derived from equation \eqref{probres}. If we denote matrix $U=(\mathbb{E}[t_j^i])_{ij}$, by using equation \eqref{expected1} and  matrix algebra, we have 
\begin{equation}
U = \sum_{i=0}^{\infty}T^i =(I-T)^{-1}.
\end{equation}
It is not difficult to verify that 
\begin{equation}
\nonumber \\
 U=(I-T)^{-1}=
 \begin{bmatrix}
 \frac{2^{h}}{2^{h}-1} & \frac{2^{{h}-1}}{2^{{h}-1}-1} &  \frac{2^{{h}-2}}{2^{{h}-2}-1}& \cdots &2 \\ 
 0 & \frac{2^{{h}-1}}{2^{{h}-1}-1} & \frac{2^{{h}-2}}{2^{{h}-2}-1}& \cdots&2\\ 
 0 & 0 & \frac{2^{{h}-2}}{2^{{h}-2}-1}&  \cdots&2\\ 
 \vdots & \vdots & \vdots & \ddots & \vdots \\ 
 0  & 0 & 0 & \cdots &  2
\end{bmatrix}.
\end{equation}
We are interested in computing the average number of vectors ($t_a$) to get to absorption starting from $k_l=0$. Hence,
\begin{eqnarray} 
\mathbb{E}_v = \mathbb{E}[t_a] &=& 
\begin{bmatrix}1 & 0 & \cdots & 0\end{bmatrix} U e \nonumber \\
&=& \begin{bmatrix}1 & 0 & \cdots & 0\end{bmatrix} (I-T)^{-1} e \nonumber \\
&=& \sum_{i=1}^{h} \frac{2^i}{2^i-1} ={h} + \sum_{i=1}^{h} \frac{1}{2^i-1} 
\end{eqnarray} 
This proves the lemma. 
 \label{prof1}
\end{proof}
\begin{rem}{\em
The optimal number of linearly independent vectors to span the  vector space is $h$. Our decentralized coded content caching strategy only requires  $h + \gamma$ coded contents to span the vector space. Since $\gamma$ is considerably smaller than $h$, then our approach provides close to optimal performance in terms of the minimum required number of caches.
 }\label{rem_opt}
\end{rem}
Lemma \ref{leme1} suggests that each UT's request can be satisfied in a smaller number of hops compared to an uncoded caching strategy. This shows that our proposed decentralized coded content caching scheme is capable of removing the inherent over-caching problem in decentralized uncoded content caching. The following theorem formalizes the result. 
\begin{thm}{\em 
If the number of popular contents is upper bounded by $M s(n)^{-1}$ i.e. if $ h = O(M {s(n)^{-1}})$, then our proposed decentralized coded content caching technique reduces the average  number of hops for decoding a content to 
  \begin{equation}
  \mathbb{E}[Y | r \in \xi_h] = 
  \Theta \left( \frac{h}{M}\right).
  \label{ex_coded}
 \end{equation}
  \label{thm_coded}
}\end{thm}
\begin{proof}
Lemma \ref{leme1} shows that in order to decode a requested content, on average $\Theta( h)$ coded contents are required. Since each UT has a cache of size $M$, Lemma \ref{leme1} shows that on average we need $\Theta( \frac{h}{M})$  UTs (hops) to be able to decode the desired content. Notice that  each individual UT does not need to separately send their coded content to the requesting UT. Each UT  combines its  encoded files with a file that it receives from previous hop and forwards it to the next hop. 
\end{proof}
Using the results of theorem \ref{thm_coded} and equations \eqref{capa} and \eqref{ex113}, the throughput capacity of coded caching can be upper bounded as follows.
\begin{corol}{\em
 If $ h = O(M {s(n)^{-1}})$, the throughput capacity of a decentralized coded content caching network is upper bounded by
  \begin{equation}
  \lambda_{\textrm{coded}}(n) = O \left(\frac{M}{h F(h) \log n} \right).
  \label{capa_coded56}
 \end{equation}
The throughput in the right hand side of equation \eqref{capa_coded56} may not be achievable due to network congestion. In the following we will find achievable network throughputs. 
 \label{corol_coded93}
}\end{corol}

\begin{thm}{\em 
The decentralized uncoded and coded content caching strategies have the order throughput capacity of  
\begin{eqnarray}
  \lambda_{\textrm{uncoded}}(n) &=& \Theta \left(\frac{1}{ F(h) \log n} \left( \frac{M}{h \log(h)   }
  \right)^2 \right), 
  \label{capa_uncoded67}\\
  \lambda_{\textrm{coded}}(n) &=& \Theta \left(\frac{1}{ F(h) \log n} \left( \frac{M}{h  } \right)^2 \right), 
  \label{capa_coded67}
 \end{eqnarray}
respectively when $h \log(h) = O (M s(n)^{-1})$.
 \label{thm_bottleneck1}
}\end{thm}
\begin{proof}{
 Clearly, the square-lets that are closer to the helper are more prone to traffic congestion. Hence, if  the achievable throughput capacity of the square-let that contains the helper is computed, this capacity is also achievable in  other square-lets. This square-let must respond to all requests which are located on average within a radius of $\Theta(h/M)$ hops away from it (or $\Theta({h \log(h)}/{M})$ hops in  uncoded caching case). Therefore, it should be able to serve on average  $ \Theta ( \log n ( {h}/{M})^2 )$ 
 (or $ \Theta ( \log n ( {h \log(h)}/{M})^2 ) $ ) requests. The probability that the popular contents are requested by UTs is $F(h)$. Thus, the average number of popular content requests from this square-let is 
 $\Theta ( \log n (  {h}/{M})^2 F(h) )$ (or $ \Theta ( \log n ( {h \log(h)}/{M})^2 F(h)) $). 

 Therefore, for the case of coded caching, this square-let should be able to deliver $ \Theta (\lambda_{\textrm{coded}}(n)$  $({h}/{M})^2 F(h) \log n  )$ contents per second and $ \Theta ( \lambda_{\textrm{uncoded}}(n) (  {h \log(h)}/{M})^2 F(h) \log n  )$ contents per second for uncoded caching. Since the network has a bandwidth of $W$ and we use TDMA scheme, each square-let can only deliver $\Theta(1)$ contents per second. Therefore, both $\lambda_{\textrm{coded}}(n) ({h}/{M})^2 F(h) \log n $
 and $\lambda_{\textrm{uncoded}}(n) ~  (  {h \log(h)}/{M})^2 F(h) \log n $ should scale as $\Theta(1)$. This proves the theorem.
}\end{proof}
Since the throughput capacity in \eqref{capa_uncoded67} (or \eqref{capa_coded67}) is achievable throughput and less than equation \eqref{capa_uncoded93} (or \eqref{capa_coded56}), then equation \eqref{capa_uncoded93} (or \eqref{capa_coded56}) cannot be achieved and is only an upper bound.
\begin{rem}
 {\em
Theorem \ref{thm_bottleneck1} shows that coded caching increases the throughput capacity of the multihop femtocaching D2D network by a factor of $\Theta((\log(h))^2)$. Since $h$ can be a function of $n$ (as will be shown in the next section), it can be concluded that coded caching can increase the throughput capacity of the network up to a factor of $\Theta((\log n)^2)$.
 }
 \label{coded_gain}
\end{rem}

\begin{rem}
 {\em 
In general, majority of the contents can be retrieved from other UTs. However, if a content cannot be retrieved from UTs on the path between the helper and the requesting UT, then the helper directly sends this content through multihop communication to the requesting UT. This can happen both in coded and uncoded caching schemes and it happens when the requesting UT is very close to the helper node.
}
\end{rem}

\section{Capacity of  networks with Zipfian content request distribution}
\label{zipfian_networks}
In this section, we compute the throughput capacity of Zipfian distribution by utilizing the results in sections \ref{uncoded_sec} and \ref{coded_sec}. To proceed, we first prove the following lemma. 
\begin{lem}{\em
 Let $\mu$ and $\eta$ be constants such that  $\mu > \eta >0$ and let $a(n)$ and $b(n)$ be two functions that scale as $\Theta(n^{\eta})$ and $\Theta(n^{\mu})$ respectively. Then, for $s>1$ we have
 \begin{eqnarray}
  \sum_{i=a(n)+1}^{b(n)} i^{-s} = \Theta (n^{-\eta (s-1)}).
  \label{mfbh}
 \end{eqnarray}
 \label{lem_zipf}
}\end{lem}
\begin{proof}
 Let $d(n) \triangleq \lfloor \frac{b(n)}{a(n)} \rfloor$. We have,
  \begin{align}
  \sum_{i=a(n)+1}^{b(n)} i^{-s} &= \sum_{j=1}^{d(n)-1} 
  \sum_{i=ja(n)+1}^{(j+1)a(n)} i^{-s} + \sum_{i=d(n) a(n)+1}^{b(n)} i^{-s}
  \nonumber \\
  &\le \sum_{j=1}^{d(n)-1} \sum_{i=ja(n)+1}^{(j+1)a(n)} i^{-s}. 
  \label{msdf_temyufbh11}
 \end{align}
 Since for $ja(n)+1 \le i \le (j+1)a(n)$ we have $ja(n) < i$, then
 \begin{align}
  \sum_{i=ja(n)+1}^{(j+1)a(n)} i^{-s} < \sum_{i=ja(n)+1}^{(j+1)a(n)} (j a(n))^{-s} 
  = a(n) (j a(n))^{-s}.
  \label{mdby_polghr11}
 \end{align}
Combining \eqref{msdf_temyufbh11} and \eqref{mdby_polghr11}, we arrive at 
  \begin{eqnarray}
  \sum_{i=a(n)+1}^{b(n)} i^{-s} < {a(n)}^{-s+1} \sum_{j=1}^{d(n)-1} j^{-s} 
< {a(n)}^{-s+1} \zeta(s),
  \label{mimulo}
  \end{eqnarray}
  where $\zeta(s)=\sum_{i=1}^{\infty} i^{-s}$ denotes the Reimann zeta function and it is a constant value 
  for $s>1$. Therefore, the upper bound is given by  
    \begin{eqnarray}
  \sum_{i=a(n)+1}^{b(n)} i^{-s} = O ({a(n)}^{-s+1}) = O (n^{-\eta(s-1)}). 
  \label{mimnuilulo}
  \end{eqnarray}
 For the lower bound of this summation, we will use an integral approximation to derive the  results. 
  \begin{align}
   \sum_{i=a(n)+1}^{b(n)} i^{-s} &\ge \int_{a(n)+1}^{b(n)} x^{-s} 
   \mathrm{d}x \nonumber \\ 
   &= \frac{(a(n)+1)^{-s+1}-{b(n)}^{-s+1}}{s-1}
   \label{mmmm_lem0}
  \end{align}
Since, $a(n)$ scales as $\Theta(n^{\eta})$ and $b(n)$ scales as $\Theta(n^{\mu})$ and $\mu > \eta$, then the first term in the right hand side of equation \eqref{mmmm_lem0} is dominant and we have 
    \begin{eqnarray}
  \sum_{i=a(n)+1}^{b(n)} i^{-s} = \Omega (n^{-\eta(s-1)}). 
  \label{mimnuilulo_lb}
  \end{eqnarray}
 This proves the lemma.
\end{proof}
As mentioned in section \ref{netmod}, we assume that $m$ is growing polynomially with $n$. Lets assume 
that $h$ which is a tiny fraction of $m$ also grows polynomially with $n$. We will now find the necessary growth rate of $h$ with $n$ to guarantee that the request probability for non-popular contents decays polynomially with $n$ with a decay rate of $\rho > 0$. In other words, we want to find the necessary growth 
rate for $h_{\rho}$ such that for constants $c_6$ and $n_1$ and for any $n > n_1$ we have
 \begin{align}
  \mathbb{P}[r > h_{\rho}] \le c_6 n^{-\rho}.
  \label{eqs_criteria}
 \end{align}
Assuming that the necessary growth rate for $h_{\rho}$ is $\tau$, i.e.$h_{\rho}$= $\Theta(n^{\tau})$, using Lemma \ref{lem_zipf}, we have 
  \begin{align}
   \sum_{i=h_{\rho}+1}^m i^{-s}= \Theta(n^{-\tau (s-1)}).
   \label{eqs_h_finds1}
  \end{align}
 This means that there exist constant $c_7$ and $n_2$ such that for any $n > n_2$ we have   
    \begin{align}
   \sum_{i=h_{\rho}+1}^m i^{-s} \le c_7 n^{-\tau (s-1)}.
   \label{eqs_h_finds2}
  \end{align}
 Since $H_{m,s} \ge 1$, we arrive at
 \begin{align}
  \mathbb{P}[r > h_{\rho}] =1-F(h_{\rho}) = \sum_{i=h_{\rho} +1}^{m} 
  \frac{i^{-s}}{H_{m,s}}  \le c_7 n^{-\tau (s-1)}.
  \label{h_ngtru}
 \end{align}
 Therefore, to ensure that equation \eqref{eqs_criteria} remains valid, it is enough to choose $c_7$ equal to $c_6$ and $\tau=\frac{\rho}{s-1}$. Hence, to guarantee that \eqref{eqs_criteria} holds, $h_{\rho}$ should scale as 
 \begin{equation}
  h_{{\rho}} = \Theta(n^{\frac{\rho}{s-1}}).
  \label{h_tau}
 \end{equation}
 Therefore, in a network where $m$ grows polynomially with $n$ with a rate of $\alpha > \frac{\rho}{s-1}$,
 if $h_{\rho}$ grows as \eqref{h_tau}, the request probability for non-popular contents decays faster than  $\Theta(n^{-\rho})$. For the rest of this section, we assume that $\rho$ is a design parameter and the helper caches contents from among the $h_{{\rho}} = \Theta(n^{\frac{\rho}{s-1}})$ most popular contents.

\begin{rem}{\em If we choose $h_{\rho}$ based on equation \eqref{h_tau} such that it satisfies $h_{\rho}  \log(h_{\rho} )=O(M s(n)^{-1})$,  equations \eqref{capa_uncoded67} and \eqref{capa_coded67} can be rewritten as 
 \begin{eqnarray}
  \lambda_{\textrm{uncoded}}^{{\rho}}(n) &=& \Theta \left( \frac{1}{\log n}    \left(
  \frac{M}{h_{\rho} \log(h_{\rho} )} \right)^2 \right), 
  \label{mmmm_uncoded} \\
  \lambda_{\textrm{coded}}^{{\rho}}(n) &=& \Theta \left( \frac{1}{\log n}    \left( \frac{M}{h_{\rho} } \right)^2 \right). 
  \label{mmmm_coded}
 \end{eqnarray}
 Equations \eqref{mmmm_uncoded} and \eqref{mmmm_coded} show that coded caching can increase the throughput capacity of Zipfian networks by a factor of $(\log(h_{\rho}))^2$ which from equation \eqref{h_tau}, it implies a factor of $(\log n)^2$ increase in throughput capacity.
 \label{mmmd_rufgsdf}
 }\end{rem}
\begin{thm}{\em
For a network with non-heavy tailed Zipfian content request distribution ($s>1$) such that the probability of content request from the base station decays polynomially with $n$ with a rate of $\rho$ and $h_{\rho}  \log(h_{\rho} )=O(M s(n)^{-1})$, then the following throughputs are achievable for the D2D network.
 \begin{eqnarray}
   \lambda_{\textrm{uncoded}}^{\textrm{Zipf},\rho}(n)  &=& \Theta \left(\frac{n^{- \frac{2\rho}{s-1}}}{(\log n)^3 } M^2 \right)
  \label{zipf_regimeI2222} \\
   \lambda_{\textrm{coded}}^{\textrm{Zipf}, \rho}(n)  &=& \Theta \left( \frac{n^{- \frac{2\rho}{s-1}}}{\log n } M^2 \right)
   \label{zipf_regimeI22222I} 
 \end{eqnarray}
 }\label{thm_maansdool}
\end{thm}
\begin{proof}{
 As mentioned earlier, in a Zipfian content request distribution with $s>1$
 to ensure that the probability of requesting non-popular contents decays polynomially
 with $n$ with a rate of $\rho$, it is enough to choose $h_{\rho}$ as in equation \eqref{h_tau}. If we use 
 this value for $h_{\rho}$ and plug it in equations \eqref{mmmm_uncoded} and \eqref{mmmm_coded}, we will 
 arrive at equations \eqref{zipf_regimeI2222} and \eqref{zipf_regimeI22222I}.
}\end{proof}

\section{Simulations}
\label{sim_sec}
The simulation results  compare the performance of  proposed decentralized coded content caching with decentralized uncoded content caching. The helper serves $n=1000$ UTs in the network with Zipfian content request probability with parameter $s=2$. In this network,  $h=100$ highly popular contents  account for more than 99\%  of the total content requests. Our simulations are carried over a cell with radius 2000 meters and for a D2D transmission range of 10 meters. Figure \ref{fig_sim_actual} shows the simulation results comparing the average number of hops  required to decode the contents in both decentralized coded and uncoded content caching schemes. The simulation results clearly demonstrate that decentralized  coded cache placement outperforms uncoded case particularly when the cache size is small which is the usual 
operating regime. For instance, with decentralized coded content caching, a cache of size 20 only requires less than 5 hops while decentralized uncoded content caching needs around 22 hops for successful content retreival. This makes  coded content caching much more practical compared to  uncoded content caching. Note that capacity is inversely proportional to the average hop counts. 
\begin{figure}[http]
    \center
      \includegraphics[scale=0.3,angle=0]{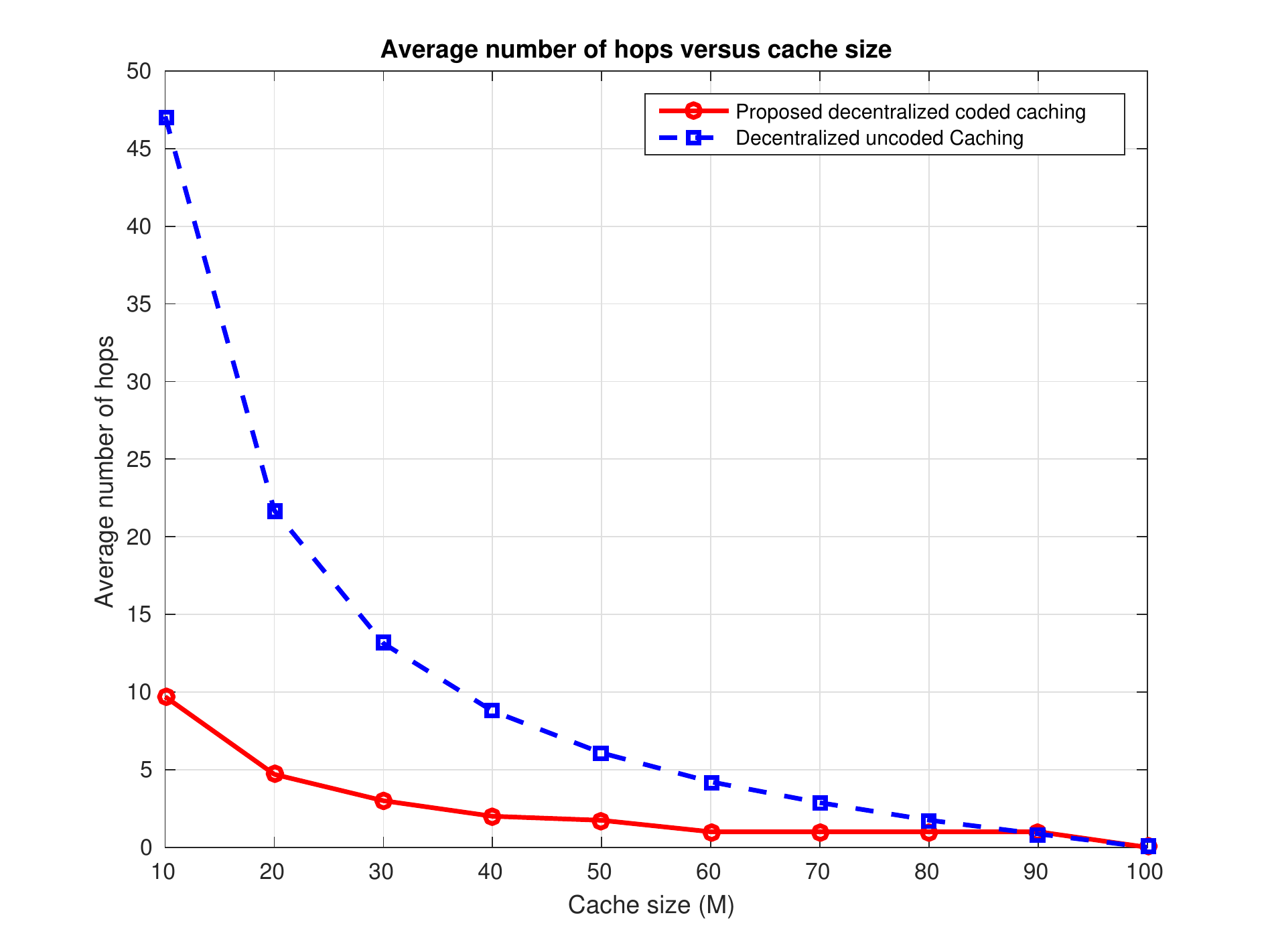}
\caption{\color{black} Simulation results for a helper serving 1000 UTs in a cell of radius 2000 meters 
with a D2D transmission range of 10 meters and a total of 100 popular contents. }
\label{fig_sim_actual}
\end{figure}

As can be seen from Figure \ref{fig_sim_actual}, for small cache sizes, coded cache placement significantly reduces the number of  hops required to decode the contents. This property is important for systems with small storage capability for UTs since large number of  hops can impose excessive delay and low  quality of service.

Figure \ref{fig_theory} compares throughput capacity of  coded content caching  with  uncoded content caching. The content popularity request distribution is Zipfian with parameter $s=2.5$. The results demonstrate significant capacity gain for decentralized coded content cache placement. The parameter $\rho = 0.75$ suggests that $h$ scales as $\Theta(\sqrt{n})$ in this plot (equation \eqref{h_tau}). Notice that in this plot, $m$ can scale as $\Theta(n^{\alpha})$ where $\alpha > 0.5$ can potentially be a large number.
As can be seen from Figure \ref{fig_theory}, for only 100 UTs and a constant cache size, throughput capacity of coded caching is 20 times higher than the throughput capacity of uncoded caching.

\begin{figure}[http]
    \center
      \includegraphics[scale=0.35,angle=0]{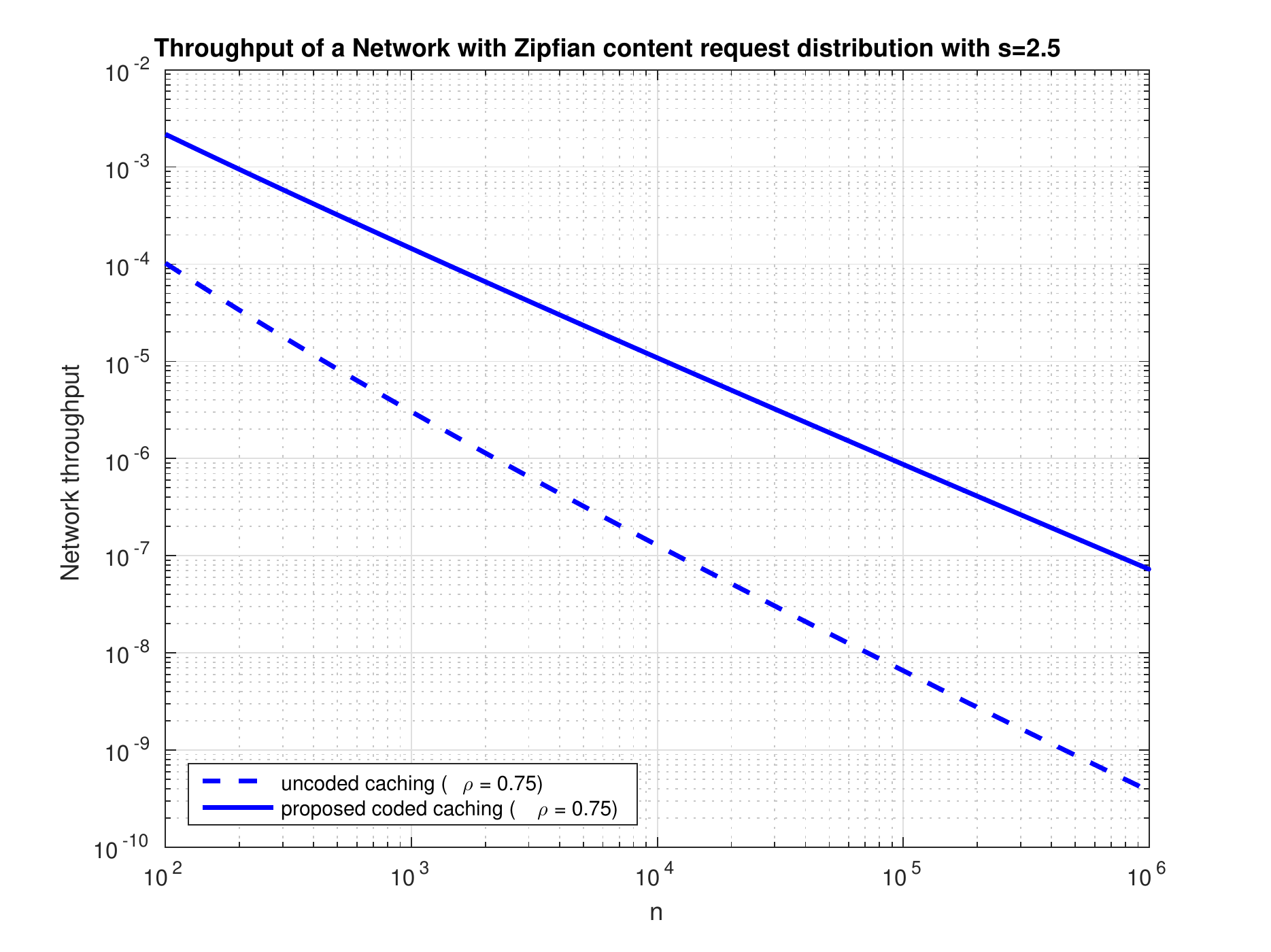}
\caption{Network throughput capacity comparison of the decentralized coded content caching and decentralized uncoded content caching schemes.}
\label{fig_theory}
\end{figure}

\section{Discussion}
\label{disc_sec}
Decentralized coded caching uses uniform random vectors in $\mathbb{F}_2^h$. This approach can be considered as special case of LT-codes \cite{DBLP:journals/tit/Shokrollahi06}. In fact, coded caching is a  random LT-code with parameters $(h, \Omega(x))$ where $\Omega(x) = \frac{1}{2^h}(1+x)^h$ is the generator polynomial \cite{DBLP:conf/focs/Luby02}. LT-codes are a practical realization of fountain codes which are proven to be very useful in erasure channels and storage systems. The uniform random LT-code in our paper is using all the possible random vectors in $\mathbb{F}_2^h$ for encoding. This allows the receiver UT to decode any  content by accessing $h+1.6067$ cache locations (on average) which is very close to the optimal value of $h$. Using other types of LT-codes, the decoding cost can be reduced but they increase the number of required cache locations to decode the contents and therefore decrease the capacity compared to uniform random LT-codes.  On the other hand, Raptor codes \cite{DBLP:journals/tit/Shokrollahi06} which are another class of fountain codes can be used to perform the decoding in constant time with slightly fewer number of required cache locations (greater than or equal to $h$). However, if $h$ is large, then the achieved throughput capacity with raptor codes is very close to our proposed technique.

The proposed approach is similar to  Random Linear Network Coding (RLNC) \cite{ho2006random}. In RLNC, each content is divided into chunks and those chunks are randomly coded and distributed in the network. In order for the user to decode the content, it must receive enough innovative packets such that it can decode those chunks. In our approach, we linearly combine different contents and only are interested in one of the contents. Therefore, in our decoding approach, we don't transmit all encoded files to the receiver. Instead, the encoding information is sent hop by hop from the helper to the UTs. Each UT needs to use that encoding information to uniquely combine its cached files with the received file from previous hop and transmit it to the next hop for more processing. One can use coded content caching to store files and then use network coding to transfer these coded contents in the network instead of simply transmitting the entire encoded files. 

\section{Conclusions}
\label{conc_sec}
In this paper, we studied the throughput capacity of cellular networks with femtocaches using decentralized uncoded and coded content cache placement for UTs. We proposed multihop communications to take away some of the communication burden from the helpers and base station and transfer it to UTs with storage capability. We proved that multihop communication together with clever use of UT cache placement strategy can increase the throughput capacity of these networks. 

We studied the case of coded caching and  proposed a near optimal decentralized coded content cache placement scheme which can increase the throughput capacity by a factor of $(\log n)^2$ over decentralized uncoded content caching. Using our proposed decentralized coded content cache placement scheme, we computed
the throughput capacity of cellular networks operating under a Zipfian content request distribution. 

The results are achieved with minimal overhead in contrast to works like \cite{jeon2015caching,jeon2015wireless} where higher overhead is required to find the content and route. In our proposed solution,  all of the requests are sent toward the helper which significantly simplifies the routing problem. However, our solution requires more computational complexity for the helpers as they should compute the appropriate coding gains and send them to the UTs along the routing path. Therefore, in our proposed solution, helpers not only require to have abundant storage capacity but also  considerable
computational capability. In fact, we are trading bandwidth with added computational complexity and storage in our decentralized coded content caching scheme. In our future work, we intend to minimize this computational complexity for the helpers. Other important issues are related to delay and security that this paper did not address and will be part of future work. 

\bibliographystyle{plain}
\bibliography{All-Papers-J4}

\begin{IEEEbiography}
[{\includegraphics[width=1in,height=1.25in,clip,keepaspectratio]{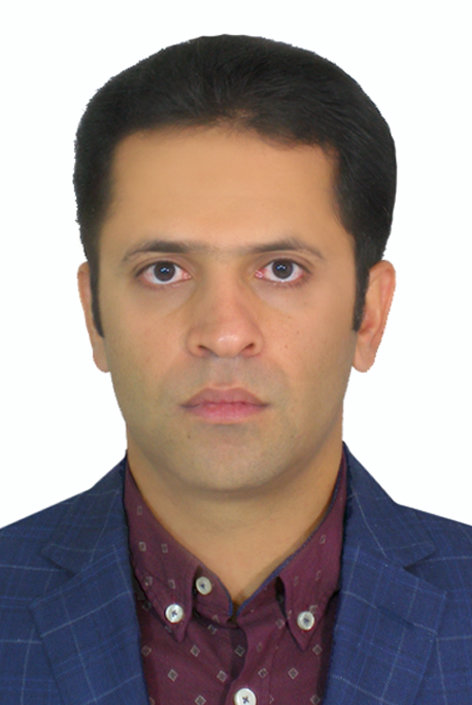}}] {Mohsen Karimzadeh Kiskani} received his bachelors degree in mechanical engineering from Sharif University of Technology in 2008. He got his Masters degree in electrical engineering from Sharif University of Technology in 2010. 
He got a Masters degree in computer science from University of California Santa Cruz in 2016. He is currently a PhD candidate in electrical engineering department at University of California, Santa Cruz. His main areas of interest include wireless communications and information theory. He is also interested in the complexity study of Constraint Satisfaction Problems (CSP) in computer science.
\end{IEEEbiography}
\begin{IEEEbiography}
[{\includegraphics[width=1in,height=1in,clip,keepaspectratio]{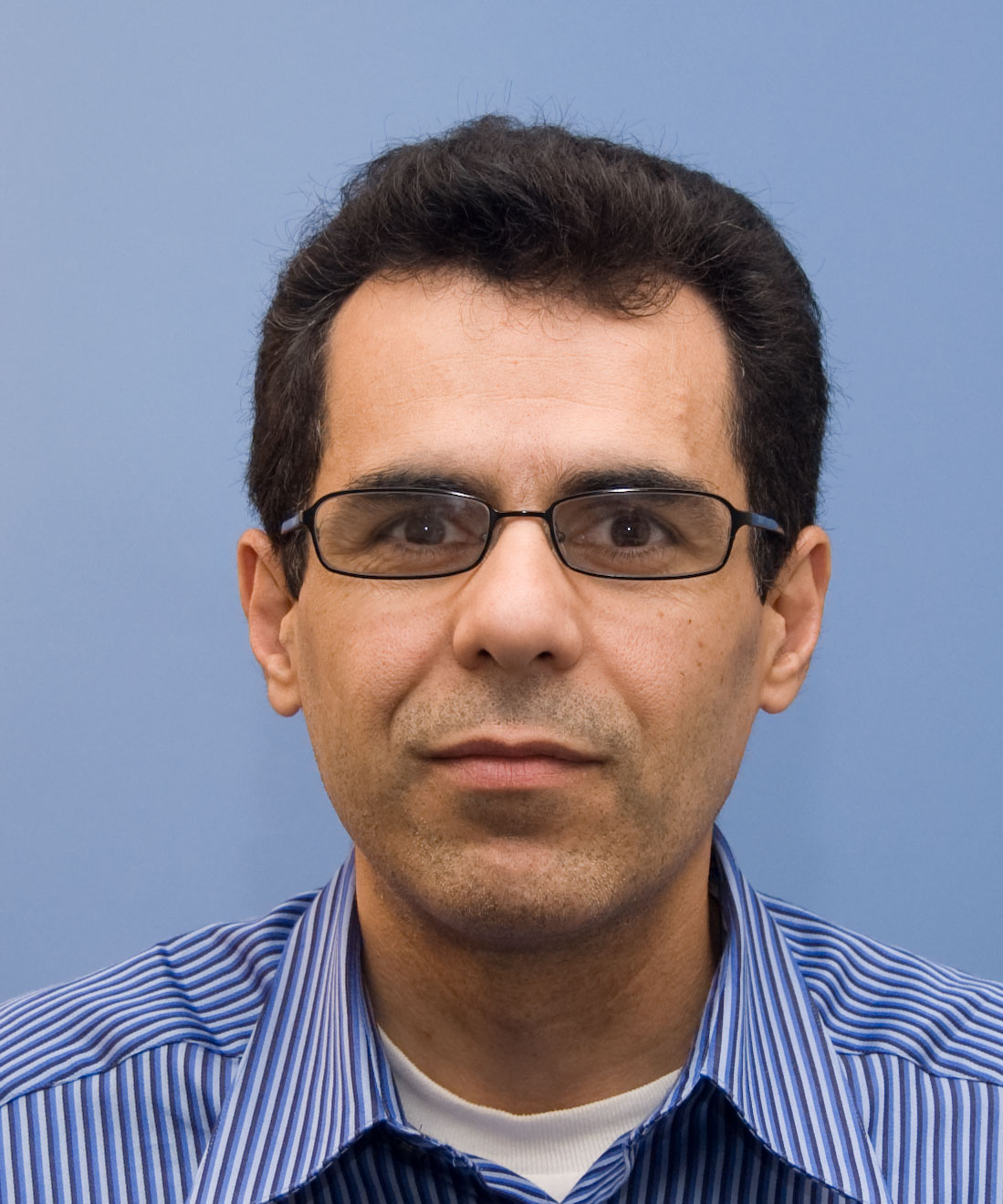}}] {Hamid Sadjadpour}
(S’94–M’95–SM’00) received the B.S. and M.S. degrees from the Sharif University of Technology, and the Ph.D. degree from the University of Southern California at Los Angeles, Los Angeles, CA. In 1995, he joined
the AT\&T Research Laboratory, Florham Park, NJ, USA, as a Technical Staff Member and later as a Principal
Member of Technical Staff. In 2001, he joined the University of California at Santa Cruz, Santa Cruz,
where he is currently a Professor. He has authored over 170 publications. He holds 17 patents. His research interests are in the general areas of wireless communications and networks. He has served as a Technical Program Committee Member and the Chair in numerous conferences. He is a co-recipient of the best paper awards at the 2007 International Symposium on Performance Evaluation of Computer and Telecommunication Systems and the 2008 Military Communications conference, and the 2010 European Wireless Conference Best Student Paper Award. He has been a Guest Editor of EURASIP in 2003 and 2006. He was a member of the Editorial Board of Wireless Communications and Mobile Computing Journal (Wiley), and the  Journal Of Communications and Networks.
\end{IEEEbiography}

\end{document}